\documentclass[a4paper]{article}

\usepackage[T1]{fontenc}
\usepackage[utf8]{inputenc}
\usepackage{enumerate}
\usepackage{algorithm2e}\IncMargin{0.5em}
\usepackage{mathtools}
\usepackage{url}
\usepackage{lmodern}
\usepackage{amssymb,amsmath,amsthm}
\usepackage{pgf,tikz}\usetikzlibrary{arrows}

\newcommand{\mapk}{\texttt{mapke5e6}}
\newcommand{\Support}{\operatorname{supp}}
\newcommand{\Supportp}{\Support^{+}}
\newcommand{\Supportm}{\Support^{-}}
\newcommand{\Supportwm}{\Support^{w-}}
\newcommand{\Supportsm}{\Support^{s-}}
\newcommand{\Supportz}{\Support^{0}}
\newcommand{\Newton}{\operatorname{newton}}
\newcommand{\moc}{\operatorname{moc}}
\newcommand{\dist}{\operatorname{dist}}
\newcommand{\sign}{\operatorname{sign}}
\newcommand{\coef}{\operatorname{coeff}}
\newcommand{\N}{\mathbb{N}}
\newcommand{\zero}{\mathbf{0}}
\newcommand{\Q}{\mathbb{Q}}
\newcommand{\R}{\mathbb{R}}
\newcommand{\Z}{\mathbb{Z}}

\newcommand{\bb}{\mathbf{b}}
\newcommand{\nn}{\mathbf{n}}
\newcommand{\pp}{\mathbf{p}}
\newcommand{\qq}{\mathbf{q}}
\newcommand{\vv}{\mathbf{v}}
\newcommand{\ttt}{\mathbf{t}}
\newcommand{\xx}{\mathbf{x}}
\newcommand{\zz}{\mathbf{z}}
\newcommand{\uu}{\mathbf{-1}}
\newcommand{\abs}[1]{\left|#1\right|}

\newtheorem{theorem}{Theorem}
\theoremstyle{definition}
\newtheorem{lemma}[theorem]{Lemma}

\begin{document}

\title{Subtropical Real Root Finding}

\author{
Thomas Sturm\\
{Max Planck Institute for Informatics}\\
{Saarbrücken, Germany}\\
\url{sturm@mpi-inf.mpg.de}}

\date{January 20, 2015}

\maketitle
\begin{abstract}
  We describe a new incomplete but terminating method for real root finding for
  large multivariate polynomials. We take an abstract view of the polynomial as
  the set of exponent vectors associated with sign information on the
  coefficients. Then we employ linear programming to heuristically find roots.
  There is a specialized variant for roots with exclusively positive
  coordinates, which is of considerable interest for applications in chemistry
  and systems biology. An implementation of our method combining the computer
  algebra system Reduce with the linear programming solver Gurobi has been
  successfully applied to input data originating from established mathematical
  models used in these areas. We have solved several hundred problems with up to
  more than 800\,000 monomials in up to 10 variables with degrees up to 12. Our
  method has failed due to its incompleteness in less than 8 percent of the
  cases.
\end{abstract}

\section{Introduction}

Our work discussed here is motivated by our studies of Hopf bifurcations
\cite{HaleKocak:91a,HairerNorsett:93a} for reaction systems in chemistry and
gene regulatory networks in systems biology, which are originally given by
systems of ordinary differential equations. Hopf bifurcations can be described
algebraically
\cite{El-KahouiWeber:00a,WangXia:05a,GatermannHosten:05a,GatermannEiswirth:05a},
resulting in one very large multivariate polynomial equation $f=0$ subject to
few much simpler polynomial side conditions $g_1>0$, \dots, $g_n>0$. For such
systems one is interested in feasibility over the reals and, in the positive
case, in at least one feasible point. It turns out that, generally,
scientifically meaningful information can be obtained already by checking only
the feasibility of $f=0$, which is the focus of this article. For further
details on the scientific background, we refer the reader to our publications
\cite{SturmWeber:08a,SturmWeber:09a,WeberSturm:11a,ErramiSeiler:11a,ErramiEiswirth:13a}.

With one of our models, viz.~\emph{Mitogen-activated protein kinase (MAPK)}, we
obtain and solve polynomials of considerable size. Our currently largest
instance $\mapk$ contains 863438 monomials in 10 variables. One of the variables
occurs with degree 12, all other variables occur with degree 5. Such problem
sizes are clearly beyond the scope of classical methods in symbolic computation.
To give an impression, the size of an input file with $\mapk$ in infix notation
is 30~MB large. \LaTeX-formatted printing of $\mapk$ would fill more than 3000
pages in this document. The MAPK model actually yields even larger instances,
which we, unfortunately, cannot generate at present, because in our toolchain
Maple cannot produce polynomials larger than 32~MB.

This article introduces an incomplete but terminating algorithm for finding
real roots of large multivariate polynomials. The principle idea is to take an
abstract view of the polynomial as the set of its exponent vectors supplemented
with sign information on the corresponding coefficients. To that extent, out
approach is quite similar to tropical algebraic geometry~\cite{Sturmfels2002a}.
However, after our abstraction we do not consider tropical varieties but employ
linear programming to determine certain suitable points in the Newton polytope,
which somewhat resembles successful approaches to sum-of-square
decompositions~\cite{DBLP:journals/tcs/PeyrlP08}.

We have implemented our algorithm in Reduce~\cite{HearnSchopf:a} using direct
function calls to the dynamic library of the LP solver
Gurobi~\cite{Gurobi-Optimization-Inc.:14a}. In practical computations on several
hundred examples, our method has failed do to its incompleteness in less than 8
percent of the cases. The longest computation time observed was around 16~s. As
mentioned above, the limiting factor at present is the technical generation of
even larger input.

In Section~\ref{SE:positive} we introduce a specialization of our method that
only finds roots with all positive coordinates. This is highly relevant in our
context of reaction networks, where typically all variables are known to be
positive. We also discuss an illustrating example in detail.
Section~\ref{SE:arbitrary} generalizes our method to arbitrary roots. In
Section~\ref{SE:practical} we discuss issues and share experiences related to a
practical implementation of our method. In Section~\ref{SE:computations} we
evaluate the performance of our method with respect to efficiency and to its
incompleteness on several hundred examples originating from four different
chemical and biological models.

\section{Finding Roots with Positive Coordinates}\label{SE:positive}

Denote $\N_1=\N\setminus\{0\}$, and let $d\in\N_1$. For $a\in\R$, vectors
$\xx=(x_1,\dots,x_d)$ of either indeterminates or real numbers, and
$\pp=(p_1,\dots,p_d)\in\N^d$, we use the notations
$a^\pp=(a^{p_1},\dots,a^{p_d})$ and $\xx^\pp=x_1^{d_1}\cdots x_d^{p_d}$. We
will, however, never consider a vector to the power of a number. Our notations
are compatible with the standard scalar product as follows:
\begin{displaymath}
  (a^{\pp})^{\qq}=(a^{p_1},\dots,a^{p_d})^{\qq}=a^{p_1q_1}\cdots a^{p_dq_d}=a^{\pp\qq}.
\end{displaymath}

Consider a multivariate integer polynomial
\begin{displaymath}
  f=\sum\limits_{\pp\in\Support(f)}\coef(f,\pp)\cdot \xx^\pp\in\Z[\xx],
\end{displaymath}
where $\coef(f,\pp)\neq0$ for $\pp\in\Support(f)$, which is called the
\emph{support} of~$f$.

\subsection{Finding a Point with Positive Value}
The \emph{Newton polytope} of
$f$ is the convex hull of $\Support(f)$. It forms a polyhedron in $\R^d$, which
we identify with its vertices, formally $\Newton(f)\subseteq\Support(f)$. The
following lemma is a straightforward consequence of the convex hull property.
\begin{lemma}\label{LE:newton}
  Let $f=\coef(f,\pp)\cdot\xx^\pp+f'\in\Z[\xx]$. Assume that
  ${\pp\notin\Newton(f)}$. Then $\Newton(f)=\Newton(f')$.\qed
\end{lemma}

For $\pp\in \Support(f)$ we define $\sign(f,\pp)=\sign(\coef(f,\pp))$. We
partition the support of $f$ as follows:
\begin{eqnarray*}
  \Support(f)&=&\Supportp(f)\mathbin{\dot\cup}\Supportm(f)\mathbin{\dot\cup}\Supportz(f),\\
  \Supportp(f)&=&\{\,\pp\in\Support(f)\mid \sign(f,\pp)>0\land \pp\neq\zero\,\},\\
  \Supportm(f)&=&\{\,\pp\in\Support(f)\mid \sign(f,\pp)<0\land \pp\neq\zero\,\},\\
  \Supportz(f)&=&\Support(f)\cap\{\zero\}.
\end{eqnarray*}
Let $\Supportp(f)=\{\pp_1,\dots,\pp_r\}$,
$\Supportm(f)=\{\pp_{r+1},\dots,\pp_s\}$, and fix any order on $\Support(f)$.
The \emph{basic LP matrix $B(f)$} is composed as follows, where the last row is
present if and only if ${\Supportz(f)\neq\emptyset}$:
\begin{displaymath}
  \renewcommand{\arraystretch}{1.25}
  B(f)=
  \left[\begin{array}{c}
          \begin{array}{c}
          \\
            B^+(f)\phantom{\vdots}\\
            \strut\\
          \hline
          \strut\\
          \strut B^-(f)\phantom{\vdots}\\
          \strut\\
            \hline
            (\zero,-1)
          \end{array}
        \end{array}\right]
=\begin{array}{c}            
   \left[\begin{array}{cccc}
           p_{11}&\dots&p_{1d}&-1\\
           \vdots&\ddots&\vdots&\vdots\\
           p_{rd}&\dots&p_{rd}&-1\\
          \hline
           p_{r+1,1}&\dots&p_{r+1,d}&-1\\
           \vdots&\ddots&\vdots&\vdots\\
           p_{s,d}&\dots&p_{s,d}&-1\\
           \hline
           0 & \dots & 0 & -1
        \end{array}\right].
\end{array}
\end{displaymath}
Considering matrices concatenations of their rows, we write this also as
$B(f)=B^+(f)\circ B^-(f)\circ(\zero,-1)^*$. Whenever we write for a given matrix
$B\in\Z^{m\times n}$ a product $N^*B$, then we implicitly agree that
\begin{displaymath}
  N^*=\left[
  \begin{array}{cccccc}
    -1 & 0 & 0 & 0 & \dots \\
    0 & 1 & 0 & 0 & \dots \\
    0 & 0 & 1 & 0 & \dots \\
    \vdots & \vdots & \ddots & \ddots & \ddots\\
  \end{array}\right]\in\Z^{m\times m}.
\end{displaymath}
That is, the multiplication $N^*B$ replaces the elements of the first row of $B$
with their additive inverses. Similarly, $\uu=(-1,\dots,-1)^T$ is generally a
column matrix of suitable length. In these terms, we are going to consider
systems
\begin{displaymath}
  N^*\cdot B(f)\cdot \xx^T\leq\uu,\quad\text{where}\quad \xx=(\nn,c)\in\R^{d+1},
\end{displaymath}
which can be rewritten as follows:
\begin{eqnarray*}
  \pp_1\nn-c&\geq&1\\
  \pp_i\nn-c&\leq&-1,\qquad i\in\{2,\dots,s\}.
\end{eqnarray*}

\begin{lemma}\label{LE:one}
  Let $f\in\Z[\xx]$. Let $\nn\in\R^d$, and let $c\in\R$. Then the following are
  equivalent:
  \begin{enumerate}[(i)]
  \item The hyperplane $H(\xx)$ defined by $\nn\xx=c$ strictly separates the
    point $\pp_1$ from $\Support(f)\setminus\{\pp_1\}$, and the normal vector
    $\nn$ is pointing from $H(\xx)$ in direction $\pp_1$. In particular,
    $\pp_1\in\Newton(f)$.
  \item There is $0<\lambda\in\R$ s.t.~$N^*\cdot
    B(f)\cdot(\lambda\nn,\lambda c)^T\leq\uu$.
  \end{enumerate}
\end{lemma}

\begin{proof}
  Assume that (i) holds. The orientation of $\nn$ is chosen such that
  $\nn\cdot\pp_1>c$ and $\nn\cdot\pp_i<c$ for $i\in\{2,\dots,s\}$. Define
  $\delta=\min_{i\in\{1,\dots,s\}}\abs{\dist(\pp_i,H)}>0$. Then
  \begin{eqnarray*}
    \pp_1\cdot\nn-c&\geq&\delta\|\nn\|,\\
    \pp_i\cdot\nn-c&\leq&-\delta\|\nn\|,\qquad i\in\{2,\dots, s\},
  \end{eqnarray*}
  and we can choose $\lambda=(\delta\|\nn\|)^{-1}$.

  Vice versa, assume that (ii) holds.
  It follows that
  \begin{eqnarray*}
    \pp_1\cdot\nn&\geq& c+1/\lambda,\\
    \pp_i\cdot\nn&\leq& c-1/\lambda,\qquad i\in\{2,\dots,s\}.
  \end{eqnarray*}
  Hence $H(\xx)$ defined by $\nn\xx=c$ is a hyperplane separating $\pp_1$ from
  $\Support(f)\setminus\{\pp_1\}$, where the distance between $H(\xx)$ and
  $\Support(f)$ is at least $\|\nn\|/\lambda>0$. Furthermore, $\nn$ is
  oriented as required in (i).
\end{proof}

\begin{lemma}\label{LE:equiv}
  Let $0\neq f\in\Z[\xx]$. Then the following are equivalent:
  \begin{enumerate}[(i)]
  \item There is $(\nn,c)\in\R^{d+1}$ s.t.~$N^*\cdot B(f)\cdot(\nn,c)^T\leq\uu$.
  \item There is $(\nn,c)\in\Q^{d+1}$ s.t.~$N^*\cdot B(f)\cdot(\nn,c)^T\leq\uu$.
  \item There is $\nn\in\Z^{d}$, $c\in\Q$ s.t.~$N^*\cdot B(f)\cdot(\nn,c)^T\leq\uu$.
  \end{enumerate}
\end{lemma}

\begin{proof}
  The existence of a real solution in (i) and a rational solution in (ii)
  coincide due to the Linear Tarski Principle: Ordered fields admit quantifier
  elimination for linear formulas~\cite{LoosWeispfenning:93a}. Given a solution
  $(n_1,\dots,n_d,c)\in\Q^{d+1}$ in (ii), we can use the principal denominator
  $m\in\N_1$ of $n_1$, \dots,~$n_d$ to obtain a solution
  $(m n_1,\dots,m n_d,m c+m-1)\in\Z^d\times\Q$ in (iii). The implication from
  (iii) to (i) is trivial.
\end{proof}

\begin{lemma}\label{LE:becomepos}
  Let $f\in\Z[\xx]\setminus\Z$. Let $(\nn,c)\in\R^{d+1}$ such that
  $N^*\cdot B(f)\cdot(\nn,c)^T\leq\uu$. Then there is $a_0\in\N$ such that for
  all $a\in\N$ with $a\geq a_0$ the following hold:
  \begin{enumerate}[(i)]
  \item
    \begin{math}
      \displaystyle
      \abs{\coef(f,\pp_1)\cdot a^{\nn\pp_1}}>\abs{\sum_{i=2}^s\coef(f,\pp_i)\cdot a^{\nn\pp_i}},
    \end{math}
  \item $\sign\bigl(f(a^{\nn})\bigr)=\sign(f,\pp_1)$.
  \end{enumerate}
\end{lemma}

\begin{proof}
  (i) From $f\notin\Z$ it follows that $\pp_1\neq\zero$. By Lemma~\ref{LE:one}
  we know $\nn\pp_1>c$ and $\nn\pp_i<c$ for $i\in\{2,\dots,s\}$. It follows that
  there is $0<\delta\in\R$ such that
  \begin{eqnarray}
    \nn\pp_1&\geq&c+\delta,\label{EQ:three_1}\\
    \nn\pp_i&\leq&c-\delta\quad i\in\{2,\dots,s\}.\label{EQ:three_2}
  \end{eqnarray}
  We are going to show that
  $a_0=\bigl\lceil\max\bigl\{2,
  \bigl(b\cdot(k-1)\bigr)^{\frac{1}{\delta}}\bigr\}\bigr\rceil$
  is a suitable choice, where
  \begin{displaymath}
    b=\abs{\coef(f,\pp_1)}^{-1}\cdot\max_{i\in\{2,\dots,s\}}\abs{\coef(f,\pp_i)}.
  \end{displaymath}

  For $a\geq a_0\geq2$ and for all $i\in\{2,\dots,s\}$, the inequalities
  (\ref{EQ:three_1}) and (\ref{EQ:three_2}) and monotony yield
  \begin{displaymath}
    a^{\nn\pp_1}\geq
    a^\delta a^c> a^\delta a^c a^{-\delta} \geq
    a^\delta a^{\nn\pp_i}\geq
    b\cdot(k-1)\cdot a^{\nn\pp_i}.\\
  \end{displaymath}
  Using the triangle inequality it follows that
  \begin{displaymath}
    a^{\nn\pp_1}>b\sum_{i=2}^s  a^{\nn\pp_i}
    \geq \abs{\coef(f,\pp_1)}^{-1}\cdot\abs{\sum_{i=2}^s\coef(f,\pp_i)\cdot a^{\nn\pp_i}},
  \end{displaymath}
  which straightforwardly implies
  \begin{displaymath}
    \abs{\coef(f,\pp_1)\cdot a^{\nn\pp_1}}>\abs{\sum_{i=2}^s\coef(f,\pp_i)\cdot a^{\nn\pp_i}}.
  \end{displaymath}

  (ii) It follows from (i) that for $a\geq a_0$ the sign of the monomial
  $\coef(f,\pp_1)\cdot a^{\nn\pp_1}$ determines the sign of $f(a^\nn)$. Since
  $a>0$, we obtain
  \begin{displaymath}
    \sign\bigl(f(a^{\nn})\bigr)= \sign\bigl(\coef(f,\pp_1)\cdot a^{\nn\pp_1}\bigr)=\sign(f,\pp_1).\qedhere
  \end{displaymath}
\end{proof}

\begin{algorithm}[t]
  \DontPrintSemicolon
  \SetAlgoVlined
  \LinesNumbered
  \SetNlSty{}{}{}
  \SetNlSkip{1em}
  \SetKwInOut{Input}{data}\SetKwInOut{Output}{result}
  \SetKwFunction{find}{find-positive}
  \SetKwFunction{solve}{lpsolve}
  \SetKwProg{fun}{function}{}{}
  \fun{\find{$f$}}{
    \Input{$f\in\Z[x_1,\dots,x_d]$}
    \Output{$\pp\in(\Q^+)^d$ or \texttt{"failed"}}
    \BlankLine
    $B^+:=B^+(f)$\;
    $B^-:=B^-(f)$\;
    $h := \texttt{\upshape "infeasible"}$\;
    \While{$h = \texttt{\upshape "infeasible"}$ and ${B^+}\neq {[~]}$}{
      $h := \solve{$B^+\circ B^-\circ(\zero,-1)^*$}$\;
      delete the first row from $B^+$\;
    }
    \If{$h = \texttt{\upshape "infeasible"}$}{
      \Return \texttt{"failed"}\;
    }
    $(\nn,c) := h$\;
    $t:=2$\;
    \While{$f(t^\nn)\leq0$}{
      $t := 2t$\;
    }
    \Return $t^\nn$\;}
  \BlankLine
  \fun{\solve{$B$}}{
    \Input{a matrix $B$ with $d+1$ columns}
    \Output{$(\nn,c)\in\Z^d\times\Q$ or \texttt{"infeasible"}}
    \BlankLine
    $\Pi := \text{LP problem given by $N^* B$ and $\uu$}$\;
    $h := \text{a solution $(\nn,c)\in\Q^{d+1}$ of $\Pi$ or \texttt{"infeasible"}}$\;
    \If{$h=\texttt{\upshape "infeasible"}$}{
      \Return \texttt{"infeasible"}\;
    }
    $m := \text{principal denominator of the coordinates of $\nn$}$\;
    $\nn := m\cdot \nn$\;
    $c := mc+m-1$\;
    \Return $(\nn,c)$\;
  }
  \caption{Functions \texttt{find-positive} and \texttt{lpsolve}\label{ALG:findpos}}
\end{algorithm}

After these preparations we can state our first subalgorithm as
Algorithm~\ref{ALG:findpos}.

\begin{theorem}[Correctness of \texttt{find-positive}]\label{TH:findposcorrect}
  Consider $f\in\Z[\xx]$. Then the following hold:
  \begin{enumerate}[(i)]
  \item The function \mbox{\tt find-positive} terminates.
  \item The function \mbox{\tt find-positive} returns either \mbox{\tt "failed"} or
    $\pp\in(\Q^+)^d$ with $f(\pp)>0$.
  \end{enumerate}
\end{theorem}

\begin{proof}
  (i) The termination of \texttt{lpsolve} follows from the existence of
  terminating algorithms for linear programming in line~15, including the
  Simplex algorithm~\cite{Dantzig:63a}, the ellipsoid
  method~\cite{Khakhiyan:79a}, and the interior point
  method~\cite{Karmakar:84a}. For the function \texttt{find-positive} itself,
  the number of iterations of the while-loop in line~3 is bounded by the number
  of rows of $B^+$, which is in turn bounded by the finite cardinality of
  $\Support(f)$. The termination of the while-loop in line~11 will be discussed
  with the correctness in (ii).

  (ii) To start with, the subroutine \texttt{lpsolve} solves the LP problem
  $\Pi$ defined in line 14 and, in the feasible case, $(\nn,c)$ in line 15 is a
  feasible point in $\Q^{d+1}$. The return value $(\nn,c)\in\Z^d\times\Q$ in
  line 21 is a feasible point for $\Pi$ as well. Its construction in lines
  18--20 follows the proof step from (ii) to (iii) in Lemma~\ref{LE:equiv}.

  Next, the while-loop in line 3 has the following loop invariants. Consider
  $$f_{(n)}=f-\sum_{i=1}^{n-1}\coef(f,\pp_i)\xx^{\pp_i}$$ before the $n$-th
  iteration:
  \begin{enumerate}[($\textrm{I}_1$)]
  \item $\Newton(f_{(n)})=\Newton(f)$,
  \item $B\bigl(f_{(n)}\bigr)=B_{(n)}^+\circ B^-\circ(\zero,-1)^*$.
  \end{enumerate}
  Invariant ($\textrm{I}_2$) is easy to see. Consider ($\textrm{I}_1$). For
  $n=1$ this is trivial. Before the $n+1$-st iteration we know that
  $h=\texttt{"infeasible"}$, which means that the LP problem given by
  \begin{displaymath}
    N^*\cdot\bigl((\pp_n)\circ B^+_{(n+1)}\circ B^-\circ(\zero,-1)^*\bigr)\quad\text{and}\quad\uu
  \end{displaymath}
  was infeasible at the $n$-th iteration. According to Lemma~\ref{LE:one} it
  follows that $\pp_n\notin\Newton(f_{(n)})$. Using Lemma~\ref{LE:newton} and
  the induction hypothesis we conclude
  $$\Newton(f_{(n+1)})=\Newton(f_{(n)}-\coef(f,\pp_n)\xx^{\pp_n})=\Newton(f_{(n)})=\Newton(f).$$

  The function \texttt{find-positive} has two possible exit points at lines 8
  and 13 corresponding to its two possible return values. Assume we are in line
  13. We have to show that $t^\nn\in(\Q^+)^d$ with $f(t^\nn)>0$. The while-loop
  in line 3 has terminated after $n$ iterations, and the if-condition in line 7
  is false. In line 9 we know by ($\textrm{I}_1$), ($\textrm{I}_2$), and
  Lemma~\ref{LE:one} that the feasible regions for
  $N^*\cdot(B^+\circ B^-\circ(\zero,-1)^*)\cdot\vv\leq\uu$ and
  $N^*\cdot B(f)\cdot\vv\leq\uu$ are identical. This allows us to use
  ${B^+\circ B^-\circ(\zero,-1)^*}$ instead of $B(f)$ for applying
  Lemma~\ref{LE:becomepos} to the original $f$, and our $\nn\in\Z^d$ has the
  property described there. In line 11, at the beginning of the $k$-th iteration
  of the while-loop we have $t=2^k$. By Lemma~\ref{LE:becomepos} we know that we
  will eventually have $t\geq a_0$ and thus $\sign(f(t^\nn))=\sign(f,\pp_n)>0$.
\end{proof}

\subsection{Finding a Zero}
We have discussed how to heuristically find $\pp\in(\Q^+)^d$ such that
$f(\pp)>0$ for our given $f\in\Z[\xx]$. On that basis Algorithm~\ref{ALG:findz}
computes $\zz\in(\bar\Q^+)^d$ such that $f(\zz)=0$, where $\bar\Q$ denotes the
algebraic closure of $\Q$.
\begin{algorithm}[t]
    \DontPrintSemicolon
  \SetAlgoVlined
  \LinesNumbered
  \SetNlSty{}{}{}
  \SetNlSkip{1em}
  \SetKwInOut{Input}{data}\SetKwInOut{Output}{result}
  \SetKwFunction{find}{find-zero}
  \SetKwFunction{findp}{find-positive}
  \SetKwFunction{compz}{construct-zero}
  \SetKwProg{fun}{function}{}{}
  \fun{\find{$f$}}{
    \Input{$f\in\Z[x_1,\dots,x_d]$}
    \Output{$\zz\in(\Q^+)^d$ or \texttt{"failed"}}
    \BlankLine

    $y:=f(\mathbf{1})$\;
    \If{$y=0$}{
      \Return $\mathbf{1}$\;}
    \If{$y>0$}{
      $f := -f$\;}
    $\pp := \findp{$f$}$\;
    \If{$\qq=\mbox{\tt "failed"}$}{
      \Return \texttt{"failed"}\;}
    $\zz := \compz{$f$, $\pp$, $\mathbf{1}$}$\;
    \Return $\zz$\;
  }
  \BlankLine
  
  \fun{\compz{$f$, $\pp$, $\qq$}}{
    \Input{$f\in\Z[x_1,\dots,x_d]$, $\pp$, $\qq\in\Q^d$}
    \Output{$\zz\in\Q^d$ or \texttt{"failed"}}
    \BlankLine

    $\bb := \pp + y\cdot(\qq-\pp)$, where $y$ is a new variable\;
    $g := f(\bb)$\;
    isolate $r\in{]0,1[}$ with $g(r)=0$\;
    $\zz := \bb(r)$\;
    \Return $\zz$\;
  }
  \caption{Functions \texttt{find-zero} and \texttt{construct-zero}\label{ALG:findz}}
\end{algorithm}

\begin{lemma}[Correctness of \texttt{construct-zero}]\label{LE:construct-zero}
  Consider $f\in\Z[\xx]$, and let $\pp$, $\qq\in\Q^d$ such that
  $f(\pp)f(\qq)<0$. Then the following hold:
  \begin{enumerate}[(i)]
  \item The function \mbox{\tt construct-zero} terminates.
  \item The function \mbox{\tt construct-zero} returns either \mbox{\tt
      "failed"} or $\zz\in\bar\Q^d$ with $f(\zz)=0$. If $\pp$, $\qq\in(\Q^+)^d$,
    then $\zz\in(\bar\Q^+)^d$.
  \end{enumerate}
\end{lemma}

\begin{proof}
  (i) The termination of \texttt{construct-zero} follows from the existence of
  terminating algorithms for univariate real root isolation including Sturm
  sequences~\cite{Sturm:1835} and more efficient algorithms
  \cite{CollinsAkritas:76a,AkritasStrzebonski:05a} based on Vincent's
  Theorem~\cite{Vincent:36a}.

  (ii) Since $f$ is continuous and $f(\pp)f(\qq)<0$, the intermediate value
  theorem guarantees the existence of $\zz\in\overline{\pp\qq}$ with $f(\zz)=0$.
  Formally, $\zz\in\bar\Q^d$ is a solution for $\xx$ of the following nonlinear
  system with indeterminates $x_1$, \dots,~$x_d$,~$y$:
  \begin{eqnarray}
    f&=&0\label{EQ:f}\\
    x_1 & = & p_1 + y\cdot(q_1-p_1)\label{EQ:x1}\\
    &\vdots&\notag\\
    x_d & = & p_d + y\cdot(q_d-p_d)\label{EQ:xd}\\
    y &>& 0\label{EQ:y0}\\
    y &<& 1\label{EQ:y1}.
  \end{eqnarray}
  In line 11, $\bb$ is assigned the vector of the right hand sides of the $d$
  equations~(\ref{EQ:x1})--(\ref{EQ:xd}). In line 12, these are plugged into the
  left hand side of equation~(\ref{EQ:f}) yielding a nonlinear univariate
  polynomial equation in $y$. Using any of the methods mentioned in (i), we
  obtain in line 13 a solution $r\in\bar\Q$ for $y$ of that equation subject to
  the constraints~(\ref{EQ:y0})--(\ref{EQ:y1}). That solution $r$ is a real
  algebraic number in some suitable representation~\cite{Mishra:93a}. In line 14
  we substitute $r$ back into the equations~(\ref{EQ:x1})--(\ref{EQ:xd}) to
  finally obtain $\xx=\zz\in\bar\Q$, also as a real algebraic number.

  Since $\zz\in\overline{\pp\qq}$, it follows from $\pp$, $\qq\in(\Q^+)^d$ that
  also $\zz\in(\Q^+)^d$.
\end{proof}

On the basis of Lemma~\ref{LE:construct-zero} the following theorem is
straightforward.

\begin{theorem}[Correctness of \texttt{find-zero}]
  Let $f\in\Z[\xx]$. Then the function \mbox{\tt find-zero} terminates and
  returns either \mbox{\tt "failed"} or $\zz\in(\bar\Q^+)^d$ with
  ${f(\zz)=0}$.\qed
\end{theorem}

When one is interested only in the \emph{existence} of a zero of $f$, then one
can, in the positive case, obviously skip \texttt{construct-zero} and exit from
\texttt{find-zero} after line 8. Notice that, in addition, one can then also
exit early from \texttt{find-positive} after line 8 in
Algorithm~\ref{ALG:findpos}.

\subsection{An Illustrating Example}
\begin{figure}[t]
  \begin{center}
    \includegraphics[width=0.65\textwidth]{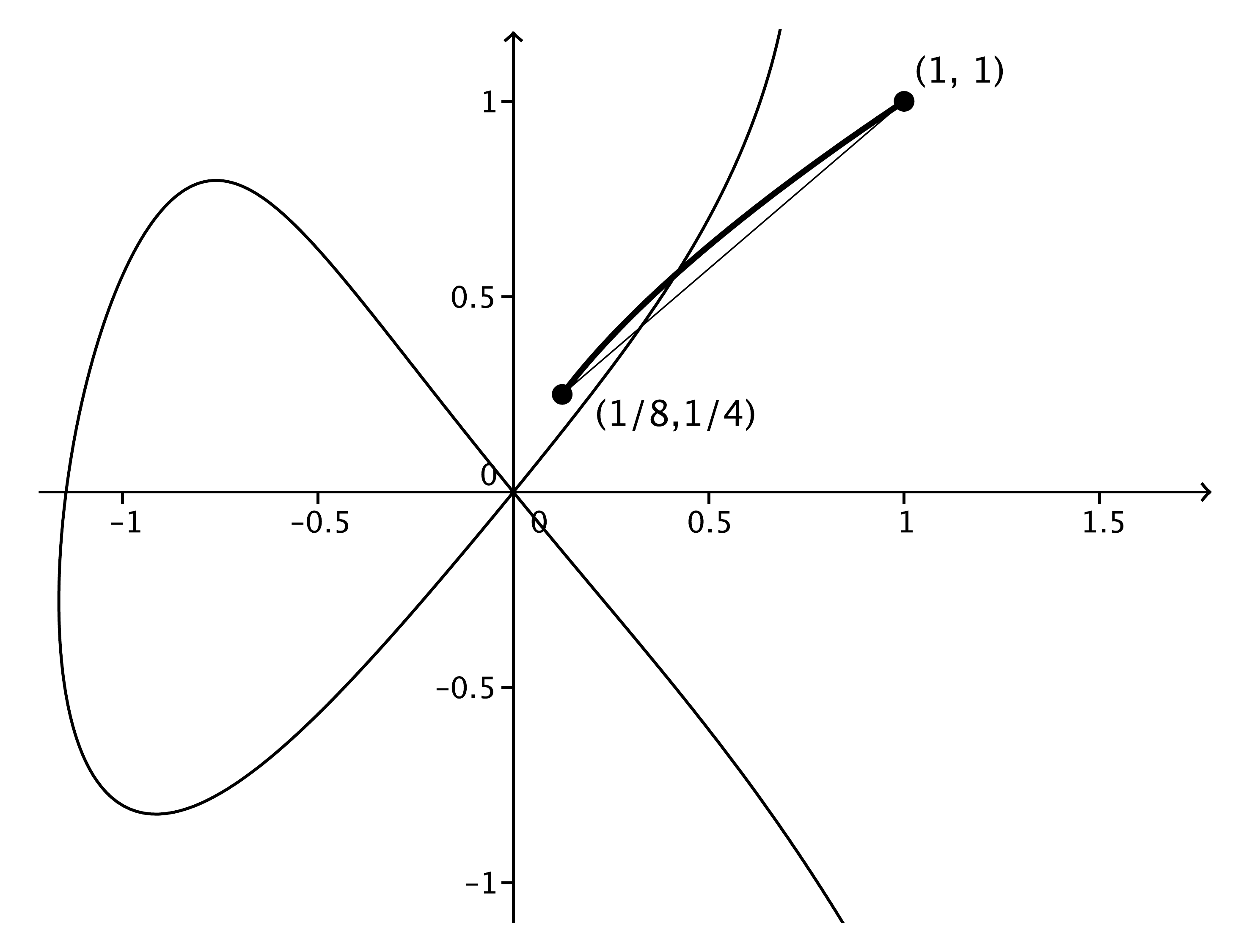}
  \end{center}
  \caption{The variety of $f=-2x_1^5+x_1^2x_2-3x_1^2-x_2^3+2x_2^2$ and the
    segment given by $t\in{[0,2]}$ of the moment curve $(t^{-3}, t^{-2})$
    corresponding to the normal vector $(-3,-2)$ of the separating hyperplane in
    Figure~\ref{FIG:tropical}.\label{FIG:real}}
\end{figure}
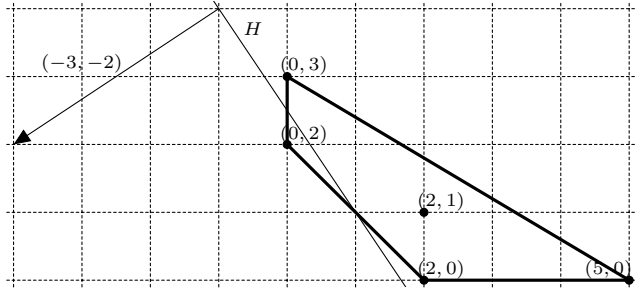
\begin{figure}[t]
  \begin{center}
    \begin{tikzpicture}[line cap=round,line join=round,>=triangle 45,x=0.9cm,y=0.9cm]
      \draw [dash pattern=on 1pt off 1pt, xstep=0.9cm,ystep=0.9cm] (-4.1,-0.1) grid (5.1,4.1);
      \clip(-4.1,-0.1) rectangle (5.1,4.1);
      \draw [line width=1.2pt] (0.,3.)-- (0.,2.);
      \draw [line width=1.2pt] (0.,2.)-- (2.,0.);
      \draw [line width=1.2pt] (2.,0.)-- (5.,0.);
      \draw [line width=1.2pt] (5.,0.)-- (0.,3.);
      \draw [domain=-4.1:5.1] plot(\x,{(-5.--3.*\x)/-2.});
      \draw [->] (-1.,4.) -- (-4.,2.);
      \begin{scriptsize}
        \draw [fill=black] (5.,0.) circle (1.5pt);
        \draw[color=black] (4.7,0.157286286293) node {$(5, 0)$};
        \draw [fill=black] (2.,1.) circle (1.5pt);
        \draw[color=black] (2.24926989947,1.15512037962) node {$(2, 1)$};
        \draw [fill=black] (2.,0.) circle (1.5pt);
        \draw[color=black] (2.24926989947,0.157286286293) node {$(2, 0)$};
        \draw [fill=black] (0.,3.) circle (1.5pt);
        \draw[color=black] (0.242390093782,3.16200018531) node {$(0, 3)$};
        \draw [fill=black] (0.,2.) circle (1.5pt);
        \draw[color=black] (0.242390093782,2.15295447296) node {$(0, 2)$};
        \draw[color=black] (-0.5,3.7) node {$H$};
        \draw[color=black] (-3,3.2) node {$(-3, -2)$};
      \end{scriptsize}
    \end{tikzpicture}
  \end{center}
  \caption{A subtropical view on ${f=-2x_1^5+x_1^2x_2-3x_1^2-x_2^3+2x_2^2}$ from
    Figure~\ref{FIG:real}. We see a hyperplane separating $(0,2)\in\Newton(f)$
    from $\Support(f)\setminus\{(0,2)\}$ together with its normal vector
    ${\nn=(-3,-2)}$.\label{FIG:tropical}}
\end{figure}
Consider $f=-2x_1^5+x_1^2x_2-3x_1^2-x_2^3+2x_2^2\in\Z[x_1,x_2]$. We apply
\texttt{find-zero} to find a point on the variety of $f$. Figure~\ref{FIG:real}
pictures the variety. We obtain $f(1,1)=-3<0$, and apply \texttt{find-positive}
to $f$.

Figure~\ref{FIG:tropical} pictures the support of $f$ and indicates the Newton
polytope. We split into $\Supportp(f)=\{(2,1),(0,2)\}$,
$\Supportm(f)=\{(2,0),(5,0),(0,3)\}$, and $\Supportz(f)=\emptyset$, and we
construct
\begin{displaymath}
  B^+=\begin{bmatrix*}
      2 & 1 & -1\\
      0 & 2 & -1
    \end{bmatrix*},\quad
  B^-=\begin{bmatrix*}
      2 & 0 & -1\\
      5 & 0 & -1\\
      0 & 3 & -1
    \end{bmatrix*}.
\end{displaymath}
Our first LP problem
\begin{displaymath}
  \begin{bmatrix*}[r]
      -2 & -1 & 1\\
      0 & 2 & -1\\
      2 & 0 & -1\\
      5 & 0 & -1\\
      0 & 3 & -1
    \end{bmatrix*}\cdot(\nn,c)^T\leq
  \begin{bmatrix*}
    -1\\
    -1\\
    -1\\
    -1\\
    -1\\
  \end{bmatrix*}
\end{displaymath}
is infeasible, which confirms the observation in Figure~\ref{FIG:tropical} that
$(2,1)\notin\Newton(f)$. Our next LP problem
\begin{displaymath}
  \begin{bmatrix*}[r]
      0 & -2 & 1\\
      2 & 0 & -1\\
      5 & 0 & -1\\
      0 & 3 & -1
    \end{bmatrix*}\cdot(\nn,c)^T\leq
  \begin{bmatrix*}
    -1\\
    -1\\
    -1\\
    -1\\
  \end{bmatrix*}
\end{displaymath}
is feasible with $\nn=(-3,-2)$ and $c=5$. Figure~\ref{FIG:tropical} shows the
corresponding hyperplane $H$ given by $-3x_1-2x_2+5=0$. It strictly separates
$(0,2)\in\Newton(f)$ from $\Support(f)\setminus\{(0,2)\}$, and its normal vector
$\nn=(-3,-2)$ is oriented towards $(0,2)$. We now know that $f(t^{-3},t^{-2})>0$
  for sufficiently large positive $t$. In fact, already
  \begin{displaymath}
    f(2^{-3},2^{-2})=f\left(\frac18,\frac14\right)=\frac{1087}{16384}.
\end{displaymath}
The relevant part of the moment curve $(t^{-3},t^{-2})$ for $t\in[1,2]$ is
pictured in Figure~\ref{FIG:real}. Since both coordinates of $\nn$ happen to be
negative, the curve will for $t\to\infty$ not extend to infinity but converge to
the origin. In particular, the curve will not leave the sign invariant region
containing $\left(\frac18,\frac14\right)$.
  
Finally, we call \texttt{construct-zero} with $\left(\frac18,\frac14\right)$ and
$(1,1)$ to solve the system
\begin{eqnarray*}
  -2x_1^5+x_1^2x_2-3x_1^2-x_2^3+2x_2^2&=&0\\
  x_1&=&\textstyle\frac18+y\cdot\left(1-\frac18\right)\\
  x_2&=&\textstyle\frac14+y\cdot\left(1-\frac14\right)\\
  y&>&0\\
  y&<&1.
\end{eqnarray*}
Dropping a positive integer denominator, we obtain the univariate polynomial
\begin{displaymath}
  \textstyle
  \bar g=-{16807} y^{5} - {12005} y^{4} - {934}
  y^{3} - {20778} y^{2} + {285} y + {1087}
\end{displaymath}
and an isolating interval $y\in{]0.2,0.3[}$. Substitution of the real algebraic
number $\bigl\langle\bar g,{]0.2,0.3[}\bigr\rangle$ into the equations for $x_1$
and $x_2$ yields an exact solution
\begin{eqnarray*}
  x_1&=&\bigl\langle 686x^5 - 78x^3 + 584x^2 - 150x - 13,{]0.32,0.33[}\bigr\rangle,\\
  x_2&=&\bigl\langle 16807x^5 - 12005x^4+ 2026x^3 + 9122x^2 - 4609x +
         323,{]0.42,0.43[}\bigr\rangle,
\end{eqnarray*}
where the intervals can, of course, be refined to arbitrary precision.
Geometrically, our solving has intersected the variety with the line segment
connecting the end points of our moment curve segment, which is also indicated
in Figure~\ref{FIG:real}.

\subsection{Why Strictly Positive Coordinates? }
In the present section, we have focused on roots with strictly positive
coordinates. This not only slightly simplifies the presentation. In fact, it is
an important feature of our algorithm to be able to perform such a directed
search.

To start with, the research presented here was originally motivated by questions
on the stability of chemical and biological reaction networks, where the
variables of the models typically are strictly positive. Our practical
computations in Section~\ref{SE:computations} are taken from those areas. For
details on the theoretical background we refer the reader to
\cite{DBLP:conf/ab/BoulierLLMU07,SturmWeber:08a,SturmWeber:09a,WeberSturm:11a,ErramiSeiler:11a,ErramiEiswirth:13a}.

Furthermore, the concept of positive feasible points is well-known from linear
programming. Techniques used there can be straightforwardly transfered to our
situation: Consider ${f\in\Z[x_1,\dots,x_d]}$. For finding zeros
$(z_1,\dots,z_d)$ with $\sign(z_i)=s_i\in\{-1,1\}$ consider
$f(s_1x_1,\dots,s_dx_d)$, for $z_1\in{]\alpha,\infty[}$ consider
$f(x_1+\alpha,x_2,\dots,x_d)$, for $z_1\in{]-\infty,\beta[}$ consider
$f(-x_1-\beta,x_2,\dots,x_d)$, and for $x_1$ unbounded consider
$f(x_1-x_1',x_2,\dots,x_d)$ introducing a new variable $x_1'$.

\section{Finding Arbitrary Roots}\label{SE:arbitrary}
\subsection{Using a Transformation}
Consider ${f\in\Z[x_1,\dots,x_d]}$. At the end of the previous section we have
addressed a technique for turning a real feasibility test based on positive
variables into a general one. Using the observation that every real number is a
difference of two positive real numbers, on introduces additional variables
$x_1'$, \dots,~$x_d'$ and transforms $f$ into $f(x_1-x_1',\dots,x_d-x_d')$. That
transformation is ubiquitous in linear programming, if not explicitly then
implicitly within the solvers.

From an efficiency point of view our procedure is clearly dominated by the LP
solving steps, where we have $d+1$ variables and $\abs{\Support(f)}$ many
constraints. Thinking in terms of state-of-the-art LP
solvers~\cite{Gurobi-Optimization-Inc.:14a,Makhorin:14a} and the Simplex method
with the option of dualization~\cite{Beale:54a,Lemke:54a}, the crucial
complexity parameter is $\min\{d,\abs{\Support(f)}\}$. With our considered
transformation the cardinality of the support increases exponentially in $d$ in
the worst case, but the number of variables only doubles from $d$ to $2d$.

Recall that our incomplete method relies on finding some $\pp\in\Newton(f)$ with
$\coef(f,\pp)>0$. We would like to point out that, doubling the dimension $d$
with the transformation, the ratio $\abs{\Newton(f)}/\abs{\Support(f)}$ will in
general increase for geometric reasons~\cite{Dwyer:88a}. Furthermore the
exponential increase of $\abs{\Support(f)}$ increases the absolute number of
candidates for a suitable $\pp$. On the other hand, the transformation does not
add points to $\Support(f)$ but exchanges it entirely. It would require either
comprehensive empirical studies or a thorough average-case analysis to make a
precise statement about the quality of the transformation in terms of
incompleteness.

\subsection{A Genuine Generalization}
We are now going to describe a generalization of the function
\texttt{find-positive} in Algorithm~\ref{ALG:findpos}, which searches for a
suitable $\pp$ not only in $\Supportp(f)$ but also in a subset of
$\Supportm(f)$. Recall that in case of success, \texttt{find-positive}
identifies the first row of $B^+$ as corresponding to the exponent vector of a
monomial with positive coefficient that dominates $f$ in the sense of
Lemma~\ref{LE:becomepos}. Then it constructs a point $\pp$ with large suitably
balanced positive coordinates. The key idea for our generalization is the
following: If the coefficient of an otherwise suitable monomial is negative but
there is at least one odd exponent in the exponent vector, then we can correct
the ``wrong'' sign of the coefficient by replacing the respective coordinate in
the constructed point $\pp$ with its additive inverse.

For $\pp=(p_1,\dots,p_d)\in\Support(f)$ define the \emph{minimal
  odd coordinate}
\begin{displaymath}
  \moc(\pp)=\min\bigl\{\,i\in\{1,\dots,d\bigr\}\bigm\vert 2\nmid p_i\bigr\},
\end{displaymath}
where $\min\emptyset=\infty$. We use the minimal odd coordinate to partition
$\Support^-(f)=\Supportwm(f)\mathbin{\dot\cup}\Supportsm$, where
\begin{eqnarray*}
  \Supportwm(f)&=&\{\,\pp\in\Supportm\mid\moc(\pp)<\infty\,\},\\
  \Supportsm(f)&=&\{\,\pp\in\Supportm\mid\moc(\pp)=\infty\,\}.
\end{eqnarray*}
The elements of $\Supportwm(f)$ are called \emph{weakly negative}. They have at
least one odd coordinate. The elements of $\Supportsm(f)$ are called
\emph{strongly negative}. They have exclusively even coordinates. We furthermore
define $B^{w-}(f)$ and $B^{s-}(f)$ corresponding to $\Supportwm(f)$ and
$\Supportsm(f)$, respectively, and we obtain
\begin{displaymath}
  B(f)=B^+(f)\circ B^{w-}(f)\circ B^{s-}(f)\circ (\zero,-1)^*.
\end{displaymath}
Consider a matrix $B$ obtained from $B(f)$ by deleting rows. Then we define
$\moc(B)=\moc(B_{11},\dots,B_{1d})$, i.e., the minimal odd coordinate of
$\pp\in\Support(f)$ corresponding to the first row of $B$.

\begin{algorithm}[t]
  \DontPrintSemicolon
  \SetAlgoVlined
  \LinesNumbered
  \SetNlSty{}{}{}
  \SetNlSkip{1em}
  \SetKwInOut{Input}{data}\SetKwInOut{Output}{result}
  \SetKwFunction{find}{find-positive-general}
  \SetKwFunction{solve}{lpsolve}
  \SetKwProg{fun}{function}{}{}
  \fun{\find{$f$}}{
    \Input{$f\in\Z[x_1,\dots,x_d]$}
    \Output{$\pp\in\Q^d$ or \texttt{"failed"}}
    \BlankLine
    $B:=B^+(f)\circ B^{w-}(f)$\;
    $B^{s-}:=B^{s-}(f)$\;
    $\mu := \infty$\;
    $h := \texttt{\upshape "infeasible"}$\;
    \While{$h = \texttt{\upshape "infeasible"}$ and $B\neq {[~]}$}{
      \If{the first row of $B$ is in $B^{w-}(f)$}{
        $\mu:=\moc(B)$
      }
      $h := \solve{$B^+\circ B\circ(\zero,-1)^*$}$\;
      delete the first row from $B$\;
    }
    \If{$h = \texttt{\upshape "infeasible"}$}{
      \Return \texttt{"failed"}\;
    }
    $(n_1,\dots,n_d,c) := h$\;
    $t:=2$\;
    $(t_1\dots,t_d):=(t^{n_1},\dots,t^{n_d})$\;
    \If{$\mu<\infty$}{
      $t_\mu:=-t_\mu$\;
    }
    \While{$f(t_1,\dots,t_d)\leq0$}{
      $t:=2t$\;
      $(t_1\dots,t_d):=(t^{n_1},\dots,t^{n_d})$\;
      \If{$\mu<\infty$}{
        $t_\mu:=-t_\mu$\;
      }
    }
    \Return $(t_1,\dots,t_d)$\;
  }
  \caption{Function \texttt{find-positive-general}\label{ALG:findposgen}}
\end{algorithm}

After these preparations we can state our function
\texttt{find-positive-general} in Algorithm~\ref{ALG:findposgen}. A
corresponding function \texttt{find-zero-general} is obtained by replacing in
\texttt{find-zero} in Algorithm~\ref{ALG:findz} the call to
\texttt{find-positive} with a call to \texttt{find-positive-general}. Everything
else remains unchanged.

For showing the correctness of \texttt{find-positive-general} we are going to
use the following variant of Lemma~\ref{LE:becomepos}:

\begin{lemma}\label{LE:becomeposgen}
  Let $f\in\Z[\xx]\setminus\Z$. Let $(\nn,c)\in\R^{d+1}$ such that
  $N^*\cdot B(f)\cdot(\nn,c)^T\leq\uu$, and let $\mu=\moc(B(f))<\infty$. Then
  there is $a_0\in\N$ such that for all $a\in\N$ with $a\geq a_0$ the following
  holds: Define $\ttt=(t_1,\dots,t_d)\in\N^d$ with $t_j=a^{n_j}$ for
  $j\in\{1,\dots,d\}\setminus\{\mu\}$ and $a_\mu=-t^{n_\mu}$. Then
  \begin{displaymath}
    \sign\bigl(f(\ttt)\bigr)=-\sign(f,\pp_1).
  \end{displaymath}
\end{lemma}

\begin{proof}
  We have $t_j=a^{n_j}>0$ for $j\in\{1,\dots,d\}\setminus\{\mu\}$,
  ${t_\mu=-a^{n_\mu}<0}$, and $p_{1\mu}$ is odd by definition of the minimal odd
  coordinate. It follows that
  \begin{equation}
    0
    < a^{\nn\pp_1}
    =-\ttt^{\pp_1}.\label{EQ:becomeposgen}
  \end{equation}
  For $i\in\{2,\dots,s\}$ we have at least $\abs{a^{\nn\pp_i}}=\abs{\ttt^{\pp_i}}$.
  This allows us to conclude from Lemma~\ref{LE:becomepos}~(i) that
  \begin{displaymath}
    \abs{\coef(f,\pp_1)\cdot \ttt^{\pp_1}}
    >\abs{\sum_{i=2}^s\coef(f,\pp_i)\cdot \ttt^{\pp_i}}.
  \end{displaymath}
  Hence $\coef(f,\pp_1)\cdot \ttt^{\pp_1}$
  determines the sign of $f(\ttt)$. Using the inequality in
  (\ref{EQ:becomeposgen}) we obtain
  \begin{displaymath}
    \sign\bigl(f(\ttt)\bigr)
    =\sign\bigl(\coef(f,\pp_1)\cdot \ttt^{\pp_1}\bigr)
    =-\sign(f,\pp_1).\qedhere
  \end{displaymath}
\end{proof}

\begin{theorem}[Correctness of \texttt{find-positive-general}]\label{TH:findposgencorrect}
  Consider $f\in\Z[\xx]$. Then the following hold:
  \begin{enumerate}[(i)]
  \item If the function \mbox{\tt find-positive} in
    Algorithm~\ref{ALG:findpos} does not fail on $f$, then
    $\texttt{\upshape find-positive-general}(f)=\texttt{\upshape find-positive}(f)$.
  \item The function \mbox{\tt find-positive-general} terminates.
  \item The function \mbox{\tt find-positive-general} returns either \mbox{\tt "failed"} or
    $\pp\in\Q^d$ with $f(\pp)>0$.
  \end{enumerate}
\end{theorem}

\begin{proof}
  (i) The function \texttt{find-positive-general} operates on
  $B^+(f)\circ B^{w-}(f)\circ B^{s-}(f)\circ (\zero,-1)^*$ while
  \texttt{find-positive} operates on $B^+(f)\circ B^{-}\circ (\zero,-1)^*$ so
  that there is possibly a different order of rows lying below $B^+(f)$.
  However, when \texttt{find-positive} does not fail, then the same feasible
  solution is found in both functions before touching anything outside $B^+$,
  and in line 10 of \texttt{find-positive-general} we have exited the while-loop
  with $\mu=\infty$. It follows that the if-conditions in lines 15 and 20 of
  \texttt{find-positive-general} are always false, and the rest of the code
  after the while-loop is computationally equivalent to the corresponding part
  of \texttt{find-positive} except for an expanded notation.

  Accordingly, a proof of parts (ii) and (iii) can be straightforwardly derived
  from the proof of (i) and (ii) of Theorem~\ref{TH:findposcorrect},
  respectively: If the function \texttt{find-positive} in Algorithm 1 does not
  fail on $f$, then there is nothing else to do. Otherwise we always reach lines
  15 and 20 with $\mu<\infty$, replace $t_\mu$ with its additive inverse, and
  apply Lemma~\ref{LE:becomeposgen} instead of Lemma~\ref{LE:becomepos} (ii),
  where we know that that $\sign(f,\pp_n)<0$.
\end{proof}

\section{Practical Issues}\label{SE:practical}
In this section we would like to discuss issues and share experiences related to
a practical implementation of our method.

One major benefit of our approach is the reduction of an algebraic problem to
linear programming (LP). Linear programming is a field with more than 50~years
of active algorithmic research, strongly driven by practical applicability and
aiming at robust implementations. Our own implementation combines the the
Codemist Standard Lisp (CSL)-based version of the computer algebra system
Reduce~\cite{HearnSchopf:a,Norman:91a,Norman:05a} with the Gurobi
Optimizer~\cite{Gurobi-Optimization-Inc.:14a}. Technically, CSL provides a
foreign function interface that allows us to dynamically load the Gurobi
C-library at runtime and call its functions from within symbolic mode Reduce
functions. Gurobi uses the Simplex algorithm. So far we have got no experience
with the use of implementations of polynomial methods for LP, like the interior
point method~\cite{Karmakar:84a}.

Gurobi uses floating point arithmetic with a limited precision. We want to
adress some issues related to this, which we consider of general inteterst,
because that floating point approach is typical for Simplex-based LP software.

On the one hand, LP solvers are quite good at controlling numerical stability.
With our comprehensive computations we have never encountered any problems with
false results due to LP rounding errors. On the other hand, in line 15 of
Algorithm~\ref{ALG:findpos} we obtain $n_1$, \dots,~$n_d$, $c$ as floats with
small rounding errors. These rounding errors do not affect correctness but cause
a subtle problem: Converting $n_1$, \dots,~$n_d$ into fractions, the GCD of
their denominators will typically be $1$ so that the principal denominator $m$
computed in line~18 becomes the very large product of those denominators.
Consequently, in line~19 we obtain our final $\nn$ with very large relatively
prime integer coordinates. This, in turn, renders infeasible the exponentiation
of increasing powers of $2$ with those integer coordinates and substitution of
the result into $f$ in line~11 of Algorithm~\ref{ALG:findpos} or in lines~14,
17, and 19 of Algorithm~\ref{ALG:findposgen}. There are two principle ways out,
which we call the \emph{pure LP approach} and the \emph{MIP approach},
respectively. Of course, the single design decisions made with these approaches
can be recombined to yield further, mixed, approaches.

\paragraph{The Pure LP Approach}
The pure LP approach tries to get along with the delivered floats. Specifically,
lines~18--20 in Algorithm~\ref{ALG:findpos} are skipped, and a floating point
vector is returned. The while-loops in line 11 of Algorithm~\ref{ALG:findpos}
and line 17 of Algorithm~\ref{ALG:findposgen} remain correct with floating point
exponents $\nn$. Later, in lines~11--12 of Algorithm~\ref{ALG:findz} it is
important to convert to rationals. In particular the substitution of floats into
a high-degree polynomial $f$ in line~12 could cause considerable numerical
instabilities. Subsequent root isolation to floating point precision in line 13
and back-substitution of the obtained floats in line~14 worked well with all our
computations.

\paragraph{The MIP Approach}
MIP stands for \emph{mixed integer (linear) programming}. Our
Lemma~\ref{LE:equiv} allows us to declare $n_1$, \dots,~$n_d$ as integers to the
LP solver right away, while $c$ remains real. As MIP is NP-hard~\cite{Karp:72a},
the MIP approach is considerably harder than the pure LP approach in terms of
theoretical complexity. In practice there are several advanced algorithms for
Simplex-based MIP solving, which rely in some way on considering an \emph{LP
  relaxation}, i.e., considering integer variables as real variables, and, in
the feasible case, trying to construct a mixed real integer feasible point on
the basis of an LP solution. The Gurobi solver specifically uses advanced
\emph{cutting plane}~\cite{Gomory:63a} methods for that construction. For the
largest problems discussed with our practical computations in
Section~\ref{SE:computations} below, we have observed factor of about 3 for MILP
solving compared to LP solving.

There is an interesting optimization with the MIP approach: Since in out
situation MIP feasibility is equivalent to LP feasibility by
Lemma~\ref{LE:equiv}, one can generally first check the latter in lines~14--15
of Algorithm~\ref{ALG:findpos}, and in the feasible case rerun for the
corresponding MIP problem. Using this strategy, there is always at most one MIP
solving step per root finding problem. Furthermore, one runs MIP solving only on
feasible instances. This excludes the really problematic cases, which are LP
feasible but not MIP feasible problems.

In rare cases one obtains integer solutions which are so large that they render
exponentiation and substitution in line~11 of Algorithm~\ref{ALG:findpos} or in
lines~14, 17, and 19 of Algorithm~\ref{ALG:findposgen} infeasible. One can
impose a suitable bound on the absolute values of the solutions, and in case of
exceeding that bound treat the problem as infeasible, and proceed to the next
candidate.

Another noteworthy optimization is the symbolic precomputation of a univariate
rational function for the while-loop in line 17 of
Algorithm~\ref{ALG:findposgen}. See Algorithm~\ref{ALG:precompute} for details.
A corresponding simpler variant, of course, works also for lines~10--12 in
Algorithm~\ref{ALG:findpos}.

For root isolation in line~13 of Algorithm~\ref{ALG:findz} we use the
\emph{Vincent--Collins--Akritas} method~\cite{CollinsAkritas:76a}. We obtain a
real algebraic number encoded by a univariate defining polynomial and an open
isolating interval, which is back-substituted in line~14, yielding a vector of
such real algebraic numbers as the final solution $\zz$.

\begin{algorithm}[t]\label{ALG:precompute}
  \DontPrintSemicolon
  \SetAlgoVlined
  \LinesNumbered
  \SetNlSty{}{}{}
  \SetNlSkip{1em}
  \setcounter{AlgoLine}{13}
  \BlankLine

  \ShowLn $(y_1,\dots,y_d)=(y^{n_1},\dots,y^{n_d})$ for a new variable $y$\;
  \addtocounter{AlgoLine}{1}
  \ShowLn \If{$\mu<\infty$}{
  \addtocounter{AlgoLine}{1}
  \ShowLn   $y_\mu=-y_\mu$\;
  }
  \addtocounter{AlgoLine}{1}
  \ShowLn $f_1:=f(y_1,\dots,y_d)\in\Z(y)$\;
  \addtocounter{AlgoLine}{1}
  \ShowLn $t := 2$\;
  \addtocounter{AlgoLine}{1}
  \ShowLn \While{$f_1(t)\leq0$}{
  \addtocounter{AlgoLine}{1}
  \ShowLn   $t:=2t$\;
  }
  \addtocounter{AlgoLine}{1}
  \ShowLn \Return $(y_1(t),\dots,y_d(t))$\;
  \BlankLine
  \caption{Code to replace lines~13--22 in
    Algorithm~\ref{ALG:findposgen}}
\end{algorithm}

\section{Some Practical Computations}\label{SE:computations}

We consider input polynomials originating from 4 different chemical and
biological models. This yields 929 instances altogether. For all of these
instances we are checking for zeros with strictly positive coordinates. It turns
out that for 640 of the instances we find $B^+=[~]$ in line 1 of
Algorithm~\ref{ALG:findpos}, which tells us that the corresponding polynomial is
positive definite (on the interior of the first hyperoctant). Running our method
on the 289 remaining instances, it fails in only 7.3~percent of the cases.
Table~\ref{TAB:bench} shows detailed information for the single models. It also
shows size (number of monomials), dimension (number of variables), and the
largest degree of an occurring variable for the respective largest instance. It
furthermore shows the maximal computation time for a single instance and the sum
of computation times.\footnote{All input and log files are available at
  \url{http://research-data.redlog.eu/arXiv/2015/subtropical/}.} All
computations have been carried out on a 2.8~GHz Xeon E5-4640 with the MIP
approach, yielding exact algebraic number solutions:

\begin{table}
  \centering  
  \small
  \begin{tabular}{l@{\hskip4pt}rrrrr}
    \hline
    & METH & OMBO & MBO & MAPK & Total\\
    \hline
    number of instances & 7 & 496 & 405 & 21 & \textbf{929}\\
    number of definite instances & 3 & 338 & 283 & 16 & \textbf{640}\\
    number of remaining instances & 4 & 158 & 122 & 5 & \textbf{289}\\
    found zero in & 4 & 144 & 107 & 5 & \textbf{260}\\
    failed on & 0 & 14 & 15 & 0 & \textbf{29}\\
    failed on (\% of remaining) & 0 & 8.9 & 12.3 & 0 & \textbf{7.3}\\
    size of largest instance  & 347 & 9787 & 9706 & 863438 & \textbf{863438}\\
    dimension of largest instance & 7 & 7 & 7 & 10 & \textbf{10}\\
    degree of largest instance  & 6 & 10  & 9 & 12 & \textbf{12}\\
    maximal time (s) & 0.16 & 4.68 & 10.00 & 15.87 & \textbf{15.87}\\
    total time (s) & 0.21 & 199.91 & 162.88 & 15.92 & \textbf{379.92}\\
    \hline
  \end{tabular}
  \caption{Statistics for our practical computations\label{TAB:bench}}
\end{table}

Notice that for our particular application the detection of definiteness by our
implementation establishes a perfect result. From that point of view, one could
argue that our method fails in only 3 percent of the cases.

\section*{Acknowledgments}
We would like to thank D.~Grigoriev, H.~Errami, W.~Hagemann, M.~Ko\v sta, and
A.~Weber for valuable discussions. A.~Norman realized a robust foreign function
interface for CSL Reduce. We are also grateful to Gurobi Optimization Inc.~and
to the \mbox{GeoGebra} Institute for making their excellent software free for
academic purposes. This research was supported in part by the German
Transregional Collaborative Research Center SFB/TR 14 AVACS and by the ANR/DFG
project SMArT.


\begin{thebibliography}{10}

\bibitem{AkritasStrzebonski:05a}
A.~G. Akritas and A.~W. Strzebonski.
\newblock A comparative study of two real root isolation methods.
\newblock {\em Nonlinear Analysis: Modelling and Control}, 10(4):297--304,
  2005.

\bibitem{Beale:54a}
E.~M.~L. Beale.
\newblock An alternative method for linear programming.
\newblock {\em Mathematical Proceedings of the Cambridge Philosophical
  Society}, 50:513--523, 1954.

\bibitem{DBLP:conf/ab/BoulierLLMU07}
F.~Boulier, M.~Lefranc, F.~Lemaire, P.-E. Morant, and A.~{\"U}rg{\"u}pl{\"u}.
\newblock On proving the absence of oscillations in models of genetic circuits.
\newblock In {\em Proceedings of the Algebraic Biology 2007}, volume
  4545 of {\em LNCS}, pages 66--80, 2007.

\bibitem{CollinsAkritas:76a}
G.~E. Collins and A.~G. Akritas.
\newblock Polynomial real root isolation using {D}escarte's rule of signs.
\newblock In {\em Proceedings of SYMSAC '76}, pages 272--275, ACM Press, 1976.

\bibitem{Dantzig:63a}
G.~B. Dantzig.
\newblock Linear programming and extensions.
\newblock Princeton University Press, Princeton, NJ, 1963.

\bibitem{Dwyer:88a}
R.~A. Dwyer.
\newblock On the convex hull of random points in a polytope.
\newblock {\em Journal of Applied Probability}, 25(4):688--699, 1988.

\bibitem{El-KahouiWeber:00a}
M.~{El Kahoui} and A.~Weber.
\newblock Deciding {H}opf bifurcations by quantifier elimination in a
  software-component architecture.
\newblock {\em Journal of Symbolic Computation}, 30(2):161--179, 2000.

\bibitem{ErramiEiswirth:13a}
H.~Errami, M.~Eiswirth, D.~Grigoriev, W.~M. Seiler, T.~Sturm, and A.~Weber.
\newblock Efficient methods to compute hopf bifurcations in chemical reaction
  networks using reaction coordinates.
\newblock In {\em Proceedings of the CASC 2013}, volume 8136 of {\em LNCS},
  pages 88--99, 2013.

\bibitem{ErramiSeiler:11a}
H.~Errami, W.~M. Seiler, T.~Sturm, and A.~Weber.
\newblock On {M}uldowney's criteria for polynomial vector fields with
  constraints.
\newblock In {\em Proceedings of the CASC 2011}, volume 6885 of {\em LNCS},
  pages 135--143, 2011.

\bibitem{GatermannEiswirth:05a}
K.~Gatermann, M.~Eiswirth, and A.~Sensse.
\newblock Toric ideals and graph theory to analyze hopf bifurcations in mass
  action systems.
\newblock {\em Journal of Symbolic Computation}, 40:1361--1382, 2005.

\bibitem{GatermannHosten:05a}
K.~Gatermann and S.~Hosten.
\newblock Computational algebra for bifurcation theory.
\newblock {\em Journal of Symbolic Computation}, 40(4--5):1180--1207, 2005.

\bibitem{Gomory:63a}
R.~Gomory.
\newblock An algorithm for integer solutions to linear programs.
\newblock In R.~L. Graves and P.~Wolfe, editors, {\em Recent Advances in
  Mathematical Programming}, pages 269--302. McGraw-Hill, 1963.

\bibitem{Gurobi-Optimization-Inc.:14a}
{Gurobi Optimization, Inc.}
\newblock {\em Gurobi Optimizer Reference Manual}, 2014.

\bibitem{HairerNorsett:93a}
E.~Hairer, S.~Norsett, and G.~Wanner.
\newblock {\em Solving Ordinary Differential Equations I. Nonstiff Problems},
  volume~8 of {\em Series in Computational Mathematics}.
\newblock Springer, 1993.

\bibitem{HaleKocak:91a}
J.~K. Hale and H.~Kocak.
\newblock {\em Dynamics and Bifurcations}, volume~3 of {\em Texts in Applied
  Mathematics}.
\newblock Springer, 1991.

\bibitem{HearnSchopf:a}
A.~C. Hearn and R.~Sch{\"o}pf.
\newblock {\em Reduce User's Manual, Free Version}, October 2014.

\bibitem{Karmakar:84a}
N.~Karmakar.
\newblock A new polynomial-time algorithm for linear programming.
\newblock {\em Combinatorica}, 4(4):373--395, 1984.

\bibitem{Karp:72a}
R.~M. Karp.
\newblock Reducibility among combinatorial problems.
\newblock In R.~E. Miller, J.~W. Thatcher, and J.~D. Bohlinger, editors, {\em
  Complexity of Computer Computations}, The IBM Research Symposia Series, pages
  85--103. Springer, 1972.

\bibitem{Khakhiyan:79a}
L.~G. Khakhiyan.
\newblock A polynomial algorithm in linear programming.
\newblock {\em Soviet Mathematics Doklady}, 20(1):191--194, 1979.

\bibitem{Lemke:54a}
C.~E. Lemke.
\newblock The dual method of solving the linear programming problem.
\newblock In {\em Naval Research Logistics Quarterly}, volume~1, pages 36--47.
  1954.

\bibitem{LoosWeispfenning:93a}
R.~Loos and V.~Weispfenning.
\newblock Applying linear quantifier elimination.
\newblock {\em The Computer Journal}, 36(5):450--462, 1993.

\bibitem{Makhorin:14a}
A.~Makhorin.
\newblock {\em GNU Linear Programming Kit}.
\newblock Department for Applied Informatics, Moscow Aviation Institute,
  Moscow, Russia, August 2014.

\bibitem{Mishra:93a}
B.~Mishra.
\newblock {\em Algorithmic Algebra}.
\newblock Texts and Monographs in Computer Science. Springer, 1993.

\bibitem{Norman:91a}
A.~C. Norman.
\newblock Codemist {S}tandard {L}isp ({CSL}) technical overview and details,
  July 1991.

\bibitem{Norman:05a}
A.~C. Norman.
\newblock Thirty years of {L}isp support for {REDUCE}.
\newblock In {\em Proceedings of the A3L 2005}. BOD, Norderstedt, Germany,
  2005.

\bibitem{DBLP:journals/tcs/PeyrlP08}
H.~Peyrl and P.~A. Parrilo.
\newblock Computing sum of squares decompositions with rational coefficients.
\newblock {\em Theor. Comput. Sci.}, 409(2):269--281, 2008.

\bibitem{Sturm:1835}
J.~C.~F. Sturm.
\newblock M\'emoire sur la r\'esolution des \'equations num\'eriques.
\newblock In {\em M\'emoires pr\'esent\'es par divers Savants \'etrangers \`a
  l'Acad\'emie royale des sciences, section Sc.~math.~phys.}, volume~6, pages
  273--318, 1835.

\bibitem{SturmWeber:08a}
T.~Sturm and A.~Weber.
\newblock Investigating generic methods to solve hopf bifurcation problems in
  algebraic biology.
\newblock In {\em Proceedings of Algebraic Biology 2008}, volume 5147 of {\em
  LNCS}, pages 200--215, 2008.

\bibitem{SturmWeber:09a}
T.~Sturm, A.~Weber, E.~O. Abdel-Rahman, and M.~{El Kahoui}.
\newblock Investigating algebraic and logical algorithms to solve {H}opf
  bifurcation problems in algebraic biology.
\newblock {\em Mathematics in Computer Science}, 2(3):493--515, 2009.

\bibitem{Sturmfels2002a}
B.~Sturmfels.
\newblock {\em Solving Systems of Polynomial Equations.}
\newblock AMS, Providence, RI, 2002.

\bibitem{Vincent:36a}
A.~J.~H. Vincent.
\newblock Sur la r{\'e}solution des {\'e}quations num{\'e}riques.
\newblock {\em Journal de Math{\'e}matiques Pures et Appliqu{\'e}es},
  1:341--372, 1836.

\bibitem{WangXia:05a}
D.~Wang and B.~Xia.
\newblock Stability analysis of biological systems with real solution
  classification.
\newblock In {\em Proceedings of the ISSAC 2005}, pages 354--361. ACM Press,
  2005.

\bibitem{WeberSturm:11a}
A.~Weber, T.~Sturm, and E.~O. Abdel-Rahman.
\newblock Algorithmic global criteria for excluding oscillations.
\newblock {\em Bull. Math. Biol.}, 73(4):899--916, 2011.

\end{thebibliography}
\end{document}